\documentclass[preprint,12pt]{elsarticle}
\usepackage{amssymb}
\usepackage{amsmath}
\usepackage[hidelinks]{hyperref}

\DeclareSymbolFont{frenchscript}{OMS}{ztmcm}{m}{n}
\DeclareMathSymbol{\Pow}{\mathord}{frenchscript}{80}  
\usepackage{xcolor}
\usepackage{iftex}
\usepackage{scalerel}
\usepackage{mathpartir}
\ifpdf
  \usepackage{underscore}         
  \usepackage[T1]{fontenc}        
\else
  \usepackage{breakurl}           
\fi
\usepackage{xspace}
\usepackage{amsthm}
\usepackage{soul}
\usepackage{stmaryrd}
\usepackage{MnSymbol}
\usepackage{amsfonts}
\usepackage{enumitem}

\usepackage{tikz}
{}
{}
{}
\newcommand{\Next}{\operatorname{%
  \tikz[baseline]{
    \draw[line width=.12ex]
      (0,.6ex) circle (.8ex);
  }}}{}
  \newcommand{\Eventually}{\operatorname{%
  \tikz[baseline]{
    \draw[line width=.12ex,line join=round]
      (0ex,.6ex) -- (.95ex,1.55ex) -- (1.9ex,.6ex) -- (.95ex,-.35ex) -- cycle;
  }}}{}
  \newcommand{\Always}{\operatorname{%
  \tikz[baseline]{
    \draw[draw=white] (-0.2ex,0ex) -- (1.7ex,0ex);
    \draw[line width=.12ex,line join=round]
      (0ex,-.2ex) -- (0ex,1.3ex) -- (1.5ex,1.3ex) -- (1.5ex,.-.2ex) -- cycle;
  }}}{}
   \DeclareMathOperator*{\bigdunion}{%
       \scalerel*{\bigcup\mathchoice%
           {\hspace*{-0.75em}\scalebox{0.7}{$\lightning$}\hspace*{0.35em}}%
           {\hspace*{-0.61em}\scalebox{0.55}{\raisebox{0.2em}{$\lightning$}}\hspace*{0.24em}}%
           {ERROR}%
           {ERROR}}{\bigcup}}
  \newcommand{\Until}{\mathbin{\,\mathcal{U}}}{}
  \newcommand{\Release}{\mathbin{\,\mathcal{R}}}{}
  
  \newcommand{\dunion}{\mathbin{%
    \cup\hspace*{-0.4em}\scalebox{0.4}{\raisebox{0.36em}{$\lightning$}}\hspace*{0.35em}}}{}

  \newcommand{\LT}{linear-time\xspace}
  \newcommand{\LTL}{Linear-time Temporal Logic\xspace}
\newcommand{\plat}[1]{\raisebox{0pt}[0pt][0pt]{#1}}     

\newcommand{\ltlthree}[0]{LTL$_3$\xspace}
\newcommand{\prepend}{\triangleright}
\newcommand{\dclose}{\mathop{\lightning}}
\newcommand{\prefixes}{\mathop{\downarrow}}
\newcommand{\extensions}{\mathop{\uparrow}}
\newcommand{\ftraces}{\Sigma^\ast}
\newcommand{\itraces}{\Sigma^\omega}
\newcommand{\traces}{\Sigma^\infty}

\newcommand{\drop}[2]{#1_{\rvert #2}}
\newcommand{\true}{\textup{\texttt{T}}}
\newcommand{\false}{\textup{\texttt{F}}}
\newcommand{\unknown}{\texttt{?}}
\newcommand{\semantics}[1]{\llbracket #1 \rrbracket}

\newcommand{\dbigunion}{\bigdunion}
\newcommand{\compl}[1]{#1^{\complement}}
\newcommand{\guar}[1]{\underline{#1}}
\newcommand{\safety}[1]{\overline{#1}}

\newcommand{\SemOp}[1]{%
  \tikz[baseline]{
    \draw[draw=white] (-0.2ex,0ex) -- (1.7ex,0ex);
    \draw[line width=.12ex,line join=round, rounded corners=0.2em]
      (0ex,-.2ex) -- (0ex,1.3ex) -- (1.5ex,1.3ex) -- (1.5ex,.-.2ex) -- cycle;
    \node at (0.75ex,0.55ex) {\tiny $#1$};
  }}

\newcommand{\Nextaif}{\operatorname{\SemOp{\Next}}}{}
\newcommand{\Negaif}{\operatorname{\SemOp{\neg}}}{}
\newcommand{\Trueaif}{\SemOp{\top}}
\newcommand{\Propaif}[1]{\left\ullcorner #1\right\ulrcorner}
\newcommand{\Untilaif}{\mathbin{\SemOp{\Until}}}
\newcommand{\Conjaif}{\mathbin{\SemOp{\land}}}
\newcommand{\Disjaif}{\mathbin{\SemOp{\lor}}}
\newcommand{\Propaifthree}[1]{\Propaif{#1}_3}

\newcommand{\weg}[1]{} 

\newtheorem{thm}{Theorem}[section]
\newtheorem{lem}{Lemma}[section]

\theoremstyle{definition}
\newtheorem{definition}{Definition}[section]

\newtheorem{counterexample}{Counterexample}[section]

\journal{Information and Computation}

\begin{document}

\begin{frontmatter}



\title{The Infinite, in Finite Time} 
%

\author[lab1]{Rayhana Amjad}
\affiliation[lab1]{organization={University of Edinburgh},
             city={Edinburgh},
             country={United Kingdom}}
\author[lab1,lab3]{Rob van Glabbeek\texorpdfstring{$^{*,}$}{}}
\author[lab2,lab1]{Liam O'Connor}
\affiliation[lab2]{organization={Australian National University},
             city={Canberra},
             country={Australia}}
\affiliation[lab3]{organization={University of New South Wales},
             city={Sydney},
             country={Australia}}
	
\begin{abstract}
Linear-time temporal properties, such as those described by \LTL, are typically modelled as sets of infinite traces. Yet, in a run-time verification context, such as when testing or monitoring a system, only a finite prefix of the system's behaviour can be observed. For some properties, these finite prefixes may be \emph{definitive}---a yes or no answer can be given without further observation. By enriching the semantics of LTL with these
	definitive prefixes, we give a proper inductive accounting of the semantics of \ltlthree, a multi-valued variant of \LTL for run-time verification applications. The semantic descriptions of \ltlthree in previous work are given only in terms of their relationship to conventional LTL\@. We show that the semantics of LTL and of \ltlthree are isomorphic. In addition, we formalise the formula progression evaluation technique, popularly used in runtime verification contexts, and show its soundness and completeness up to finite traces with respect to our semantics. 
Then, we turn to \LT properties more generally: using our theory of definitive prefixes, we re-prove the well-known safety-liveness decomposition theorem, and reconstruct the topology of infinite traces. We define \emph{monitorability} for properties, providing neat topological characterisations for various monitorability classes, and arrange them into a hierarchy. All of our definitions and proofs are mechanised in Isabelle/HOL.
\end{abstract}

\begin{keyword}
Linear-time Temporal Logic \sep Partial Traces \sep Semantics \sep Formula Progression \sep Safety Properties \sep Liveness Properties \sep Guarantee Properties \sep Topology \sep Monitorability

\end{keyword}

\end{frontmatter}
$^*${\small\it Supported by Royal Society Wolfson Fellowship RSWF\textbackslash R1\textbackslash 221008}

\section{Introduction}
In the specification, verification and modelling of reactive and concurrent systems, \LT properties are a recurring theme. Linear-time properties are sets of \emph{behaviours}: completed,
infinite \emph{traces}---sequences of states---describing the execution of a system over time. Logics for the specification of such properties, such as \LTL (LTL)~\cite{ltl} are well-studied and widely used, and a variety of algorithmic techniques for model checking and verification exist~\cite{vardiwolper}. In the context of \emph{run-time} verification, testing, or monitoring, however, semantics in terms of behaviours appears insufficient, as the full infinite behaviour of the system cannot be observed: we can only observe a finite sequence of states.

In light of this, we might consider switching to a semantics which includes finite traces
as models of temporal properties. There are several variants of LTL which include finite traces as models, the oldest of which, called LTL$\!f$, is commonly attributed to Pnueli\footnote{Such logics are found in many early papers on LTL with Pnueli as a
coauthor such as Lichtenstein et al.~\cite{fltl}, but Manna {\&} Pnueli~\cite{ltlsafety}, which
is usually cited, does not mention finite traces at all.}. These logics concern finite or infinite \emph{completed} traces, but this is also not suitable in the context of run-time verification, as our finite observations are not completed traces, but finite \emph{prefixes} of infinite behaviours. In other words, the behaviour of the system is still infinite, it is only our \emph{observation} of the behaviour that is finite.

Therefore, in this paper we advocate remaining with infinite trace semantics even in a run-time verification context. If we are monitoring for the property ``eventually, a good state is reached'', then our observation can stop the moment a good state is reached, as the property has been satisfied already and no further observation is needed. Similarly, if we a monitoring for the property ``the bad state is never reached'', our observation can stop as soon as the bad state is reached, as the property has been irrevocably violated. Thus, there are properties for which \emph{definitive} answers---satisfaction or violation---can be determined after observation of just a finite prefix. We call such prefixes \emph{definitive prefixes}, and sets of infinite traces and their definitive prefixes \emph{definitive sets}. We introduce our theory of definitive prefixes and sets in Section~\ref{sec:dps}, including a proof that definitive sets are determined uniquely by their infinite traces, and thus that definitive sets and \LT properties are isomorphic.

Bauer et al.~\cite{ltl3tosem}  describe a multi-valued variant of LTL for partial traces
called {\ltlthree} that distinguishes between those formulae that can be definitively said to be true or false from just the partial trace provided, and 
those formulae which are indeterminate, requiring further states to evaluate definitively. As we shall see in Section~\ref{sec:ltl}, the semantics of \ltlthree in the literature are given only in terms of conventional LTL, and Bauer et al.~\cite{bauercomparing} further claim that \ltlthree \emph{cannot} 
be given an inductive semantics, but this claim applies only to the intensional kind of semantics typically given for LTL. In Section~\ref{sec:sems} we give a compositional, inductive semantics for \ltlthree, in terms of families of 
definitive sets. By virtue of our isomorphism, we further show that the semantics of conventional LTL and of \ltlthree correspond directly. This correspondence is so direct that \ltlthree can be viewed merely as a different presentation of 
	conventional LTL: They share the same syntax, and their meanings are related by a simple transport.

In Section~\ref{sec:fp} we turn to \emph{formula progression}, a popular technique for evaluating formulae against a finite trace where the formula is evaluated state-by-state, in a style reminiscent of operational semantics or the Brzozowski derivative. Bauer {\&} Falcone~\cite{fpltl3fm} claim without proof that formula progression yields an equivalent semantics to \ltlthree. In this paper, we make this statement formally precise, and prove soundness and completeness (modulo a sufficiently powerful simplifier) of the formula progression technique with respect to our semantics.

In Section~\ref{sec:properties}, we investigate various classes of properties and formally describe their monitorability characteristics. Based on our definitive 
prefixes, we define the \emph{guarantee kernel} of a property, which, when viewed as an interior operator gives a topology equivalent to that of Alpern {\&} Schneider~\cite{alpernschneider}. This enables us to re-prove their classic result, that every property is the intersection of a safety and liveness property, by other means. We then formally characterise various notions of monitorability, including that of Bauer et al.~\cite{bauercomparing} and of Pnueli {\&} Zaks~\cite{PnueliZaks}, and arrange them into a hierarchy.

Finally, in Section~\ref{sec:disc}, we relate our work to other characterisations of traces and properties, as well as to other multi-valued variants of LTL\@.  All of our work has been mechanised in the Isabelle/HOL proof assistant, proofs of which are available for download~\cite{ourrepo,ourproofs}. The mechanised proof scripts are intended as a comprehensive library, and, being included in Isabelle's Archive of Formal Proofs, are maintained to ensure they are always compatible with the current stable release of Isabelle.
An extended abstract of (part of) this paper appeared as \cite{EPTCS412.4}. The present paper contains significant additional work expanding on monitorability for linear-time properties, providing topological characterisations for various monitorability classes and demonstrating a hierarchy between them.

\section{Definitive Prefix Sets}
\label{sec:dps}
In this section, we develop a theory of \emph{definitive prefixes}, which we will use in Section~\ref{sec:sems} to give a semantics to \ltlthree, and in Section~\ref{sec:properties} to define various classes of properties.

Given a set of \emph{atomic propositions}, the directly observable properties of a state of a system are the atomic propositions that hold in that state. For this reason we abstractly model a state simply as a set of atomic propositions. Let $\Sigma$ be a set of such states.
A \emph{trace} $t$ may be finite (in $\ftraces$) or infinite (in $\itraces$). We denote the set of all traces, i.e.\ $\ftraces \cup \itraces$, as $\traces$. We sometimes address a sequence $a_1 a_2 \cdots$ of atomic propositions as a trace; formally this refers to any sequence $\sigma_1 \sigma_2 \cdots \in \traces$ with $a_i \in \sigma_i$ for all relevant $i$. Two traces $t$ and $u$ may be concatenated in the obvious way, written as $tu$. If $t$ is infinite, then $tu = t$. The empty trace is denoted $\varepsilon$.
We denote the set of \emph{prefixes} of a trace $t$ as $\prefixes t$:
\[\prefixes t \triangleq \{ u \mid \exists v \in \traces.\ t = uv \}\]
We also generalise this notation to sets, so $\prefixes X$ is the set of all prefixes of traces in $X$. The set of all \emph{extensions} of a trace $t$ is likewise written as $\extensions t$:
\[
\extensions t \triangleq \{ tu \mid u \in \traces \}
\]
The set of \emph{definitive prefixes} of a set of traces $X$ is written $\dclose X$. This is the set of all traces for which all extensions are a prefix of a trace in $X$.
\[
\dclose X \triangleq \{ t \mid \extensions t \subseteq \prefixes X \}
\]
Equivalently, $t$ is a definitive prefix of $X$ iff $X$ contains all infinite extensions of $t$: $\dclose X = \{ t \mid \extensions t \cap \Sigma^\omega \subseteq X \}$.
Intuitively, this means that $\dclose X$ contains all those traces for which reaching $X\cap\Sigma^\omega$ is in some way \emph{inevitable}, even if it hasn't happened yet. Definitive prefixes are therefore similar to the notion of \emph{good} and  \emph{bad} prefixes from Kupferman {\&} Vardi~\cite{kupfermanvardi}, just without any moral judgement (see Section~\ref{sec:goodbad}). 

 A set $X$ of traces is called \emph{definitive} iff $X = \dclose X$. Let $\mathcal{D} \subseteq \Pow(\traces)$ denote the set of all definitive sets.

For example, the set $\{ t \in \traces \mid \exists uv.\ t = u\textsf{a}v\}$ (i.e., all traces that contain the state $\mathsf{a}$) is definitive, as it contains all its definitive prefixes (i.e., those prefixes that contain $\mathsf{a}$).

For any set of traces $X$, we have the following straightforwardly from the definitions:
\begin{itemize}
\item All definitive prefixes are prefixes, i.e.\ $\dclose X \subseteq \prefixes X$.
\item The set $\dclose X$ itself is definitive, i.e.\ $\dclose \dclose X = \dclose X$.
\item Any extension of a definitive prefix is also a definitive prefix, that is, $ \forall t \in \dclose X. \uparrow t \subseteq \dclose X$.
\item The definitive prefix operator $\lightning$ distributes over intersection, that is, $\lightning (\bigcap_{i \in I} X_i) = \bigcap_{i \in I} \lightning X_i$.
\end{itemize}
The sets $\emptyset$ and $\traces$ are both definitive, and the definitive sets are closed under intersection, i.e., for a set of definitive sets $S$, $\lightning \bigcap S = \bigcap S$. This follows from the distributivity theorem above. The definitive sets are not closed under union, however. To see why, consider when $\Sigma = \{\texttt{A}, \texttt{B}\}$ and the set $X_\texttt{A}$ contains all traces starting with $\texttt{A}$ and the set $X_\texttt{B}$ contains all traces starting with $\texttt{B}$. The sets $X_\texttt{A}$ and $X_\texttt{B}$ are both definitive, but their union is not: neither $X_\texttt{A}$ nor $X_\texttt{B}$ contain the empty trace $\varepsilon$, but $ \varepsilon \in \dclose (X_\texttt{A} \cup X_\texttt{B})$, as all extensions of $\varepsilon$ (i.e.\ all non-empty traces) must begin with either $\texttt{A}$ or $\texttt{B}$.

\subsection{Lattice Properties}
Define the \emph{definitive union}, written $\dbigunion S$ or $X \dunion Y$ in the binary case, as merely the definitive prefixes of the union: 
$$
\dbigunion S \triangleq \dclose \bigcup S
$$
\begin{thm}The definitive union gives least upper bounds for definitive sets ordered by set inclusion, i.e., for a set $S \subseteq \mathcal{D}$ of definitive sets:
\begin{itemize}
\item For all $X \in S$, $X \subseteq \bigdunion S$.
\item If there is a definitive set $Z$ such that $\forall X \in S.\ X \subseteq Z$, then $\bigdunion S \subseteq Z$.
\end{itemize}
\end{thm}
\begin{proof}Follows from definitions.	
\end{proof}

\noindent Thus, the definitive sets $\mathcal{D}$ ordered by set inclusion form a complete lattice, where the supremum is the definitive union, the infimum is the intersection, the greatest element $\top$ is $\traces$ and the least element $\bot$ is $\emptyset$.

\subsection[Isomorphism to LTL Properties]{Isomorphism to Linear-time Temporal Properties}

\newcommand{\Prop}{\mathsf{Pr}}
\newcommand{\Definitives}{\mathsf{Df}}
\begin{thm}\label{thm:bijection}Define the lower adjoint $\Prop : \mathcal{D} \rightarrow \Pow(\itraces)$ as $\Prop(X) = X \cap \itraces$, and the upper adjoint $\Definitives : \Pow(\itraces) \rightarrow \mathcal{D}$ as $\Definitives(P) = \dclose P$. We have, for any definitive set $X$ and \LT temporal property $P$: $$ \Prop(X) = P\ \text{if and only if}\ X = \Definitives(P)$$\end{thm}
\begin{proof} Proving each direction separately:
\begin{description}
	\item[$\implies$] \!It suffices to show that $\Definitives(\Prop(X)) = X$, i.e.\ $ \lightning(X \cap \itraces) = X$. Because any\linebreak[3] extension of a definitive prefix is also a definitive prefix, and all traces in $X$ are definitive prefixes, every finite trace in $X$ must be a prefix of some infinite trace in $X$. Thus, the infinite traces in $X$ alone are sufficient to describe all their definitive prefixes, i.e., all traces in $X$.
	\item[$\impliedby$] \!It suffices to show that $\Prop(\Definitives(P)) = P$, i.e.\ $\lightning P \cap \itraces = P$. Recall that the definitive prefixes of $P$ are all those traces $t$ for which all extensions of $t$ are a prefix of a trace in $P$. For an infinite trace $t \in P$, 
	the set of extensions $\extensions t$ is just $\{ t \}$, which is surely contained in $\prefixes P$. Therefore, for a \LT temporal property $P$ (consisting only of infinite traces), we can conclude $P \subseteq \lightning P$. In fact, $\lightning P$ consists of all the (infinite) traces of $P$ as well as possibly some finite prefixes of these. Hence $\lightning P \cap \itraces = P$.
\popQED
\end{description}	
\end{proof}
\begin{thm}
	$\Prop$ (and likewise for $\Definitives$) is monotone and preserves least upper and greatest lower bounds, i.e:
	\begin{itemize}
		\item If $A \subseteq B$ then $\Prop(A) \subseteq \Prop(B)$
		\item $\Prop(\bigcap_{i \in I} X_i) = \bigcap_{i \in I} \Prop(X_i)$
		\item $\Prop(\dbigunion_{i \in I} X_i) = \bigcup_{i \in I} \Prop(X_i)$
	\end{itemize}
\end{thm}
\begin{proof}
The first two statements follow directly from definitions. Preservation of least upper bounds requires more finesse. As we have already seen in the proof of Theorem~\ref{thm:bijection}, for any set of traces $S$, $\Prop(\lightning S) = \Prop(S)$. This means that $\Prop(\dbigunion_{i \in I} X_i) = \Prop(\bigcup_{i \in I} X_i) = \bigcup_{i \in I} \Prop(X_i)$ as required.
\end{proof}
\noindent These theorems say that $(\Prop, \Definitives)$ forms a \emph{lattice isomorphism} between definitive sets and \LT temporal properties.

\section[LTL]{\LTL}
\label{sec:ltl}

\begin{figure}
\vspace{-2ex}
\begin{center}
$$
\begin{array}{llcl}
	\text{Formulae}&\varphi, \psi & ::= & \top \mid a \\
	& & \mid & \neg \varphi \\
	& & \mid & \varphi \land \psi \\
	& & \mid & \Next \varphi \\
	& & \mid & \varphi \Until \psi \\
\text{Atomic propositions} & a & \in & A \\
\text{States} & \Sigma & = & \Pow(A) \\
\text{Traces} & t, u & \in & \traces \\[1em]
\end{array}
$$
Abbreviations:
$$
\begin{array}{lcl}
	\bot & \triangleq & \neg \top \\
	\varphi \lor \psi & \triangleq & \neg (\neg \varphi \land \neg \psi) \\
	\Eventually \varphi & \triangleq & \top \Until \varphi \\
	\varphi \Release \psi & \triangleq & \neg (\neg \varphi \Until \neg \psi)\\
	\Always \varphi & \triangleq & \bot \Release \varphi
\end{array}
$$
\end{center}
\caption{Syntax of LTL}	
\label{fig:syntax}
\end{figure}

Figure~\ref{fig:syntax} describes the syntax of LTL formulae and adjacent definitions. LTL extends propositional logic over states with temporal operators to produce a logic over traces --- sequences of states.

Our formulation takes conjunction, negation, atomic propositions and the temporal operators \emph{next} ($\Next$) and \emph{until} ($\Until$) as primitive, with disjunction and the temporal operators \emph{eventually} ($\Eventually$), \emph{always} ($\Always$) and \emph{release} ($\Release$) derived from these primitives. 
\begin{figure}
\begin{displaymath}
\begin{array}{lcl}
    t \models \top \\
	t \models a & \text{iff} & a \in t_0 \\
	t \models \neg \varphi & \text{iff} & t \nmodels \varphi\\
	t \models \varphi \land \psi & \text{iff} & t \models \varphi\ \text{and}\ t \models \psi\\
	t \models \Next \varphi & \text{iff} & \drop{t}{1} \models \varphi \\
	t \models \varphi \Until \psi & \text{iff} & \text{there exists}\ i\ \text{s.t.}\ \drop{t}{i} \models \psi\ \text{and}\ \forall j < i.\ \drop{t}{j} \models \varphi
\end{array}	
\end{displaymath}
\caption{Semantics of conventional LTL}	
\label{fig:ltlsemantics}
\end{figure}

Figure~\ref{fig:ltlsemantics} gives the semantics of conventional LTL, as a satisfaction relation
whose models are infinite traces. Here $t_0$ denotes the first state of a trace $t$. As we only include future temporal operators, we can advance to the future by dropping initial prefixes from the trace. The notation $\drop{t}{n}$ denotes the trace $t$ without the first $n$ states. If $n$ is greater than the length of $t$, the result of $\drop{t}{n}$ is the empty trace $\varepsilon$. Note that our \emph{until} operator ($\Until$) is \emph{strong}, in that $\varphi \Until \psi$ requires that $\psi$ eventually becomes true at some point in the trace. 

Bauer et al.~\cite{ltl3tosem} describe \ltlthree as a \emph{three-valued} logic that interprets LTL formulae on finite prefixes to obtain a truth value in $\mathbb{B}_3 = \{\true,\false,\unknown\}$. For a formula $\varphi$ and a finite prefix $t$, the truth value $\true$ indicates that $\varphi$ can be definitively established from $t$ alone, whereas $\false$ indicates that $\varphi$ can be definitively refuted from $t$ alone. The third value $\unknown$ indicates that the formula $\varphi$ can neither be established nor refuted from $t$ alone:
$$
[t \models_3 \varphi] = \begin{cases} \true & \text{if}\ \forall u \in \itraces.\ tu \models \varphi \\
    \false & \text{if}\ \forall u \in \itraces.\ tu \nmodels \varphi  \\
    \unknown & \text{otherwise}\end{cases}
$$%
Because the truth value $\unknown$ indicates merely that neither $\true$ nor $\false$ apply, \ltlthree can be better understood as a two-valued \emph{partial logic}~\cite{partiallogic}, where $\unknown$ indicates the \emph{absence} of a truth value. The logic is partial in the sense that its evaluation is a partial function from possibly-finite prefixes to $\{\true,\false\}$. The value $\unknown$ is not a third semantic state, but a computational placeholder for when a prefix is not contained within a property's definitive set, and thus no truth value has yet been assigned. In this view, \ltlthree only gives truth values when the trace is \emph{definitive}, i.e., when the answer given will not change regardless of how the trace is extended.

Bauer et al.~\cite{ltl3tosem} note that this presentation of  \ltlthree is inherently non-inductive, i.e., the answer given for a compound formula cannot be produced by combining the answers for its components.  To see why, consider the formula $\varphi = \Eventually a \lor \Eventually \neg a $. Using the semantics above, we have $[\varepsilon \models_3 \varphi ] = \true$ but each component of $\varphi$ produces no definitive answer for the empty trace, i.e.\ $[\varepsilon \models_3 \Eventually a] = [\varepsilon \models_3 \Eventually \neg a] = \unknown$. Likewise both components of the formula $\psi = \Eventually b \lor \Eventually \neg c$ produce the answer $\unknown$ for the empty trace, but unlike $\varphi$, we have $[ \varepsilon \models_3 \psi] = \unknown$. Therefore, there is no way to combine two $\unknown$ answers in a disjunction that produces correct answers for both $\varphi$ and $\psi$. Because of this, Bauer et al.~\cite{ltl3tosem} claim that inductive semantics are \emph{impossible} for \ltlthree. As we shall see, however, this claim applies only to the multi-valued semantics defined above. In our development, which associates sets of traces to each formula, the semantics for a given formula can indeed be compositionally constructed from the semantics of its components.

\section[Semantics of LTL and LTL3]{Semantics of LTL and \ltlthree}
\label{sec:sems}

\subsection{Answer-indexed Families}

We give a semantics to \ltlthree by compositionally assigning to each formula
an \emph{answer-indexed family} of definitive sets. In general, an \emph{answer-indexed family} is a function that, given an answer (e.g.~a value in $\mathbb{B} = \{ \true, \false \}$), produces a set of models (depending on the logic, this could be a set of states, a definitive set, a linear-time temporal property, etc.). This set contains all those models which produce the given answer for the formula in question. In this way, we invert the traditional presentation of multi-valued logics, where the truth value is the output of a satisfaction function, and instead take the desired answer $a$ as an \emph{input} and produce models as an \emph{output}: $a = [\sigma \models \varphi]$ becomes $\sigma \in \semantics{\varphi}\ a$. An answer-indexed family for conventional, infinite-trace LTL therefore produces a linear-time temporal property as an output:
\[\Phi,\Psi \; \in \;  \mathbb{B} \rightarrow \Pow (\itraces)\]
whereas an answer-indexed family for \ltlthree produces a definitive set as an output:
\[\Phi,\Psi\; \in\; \mathbb{B} \rightarrow \mathcal{D}\]
To begin with, we define an alternative semantics for \emph{conventional} LTL in terms of answer-indexed families of linear-time temporal properties. 
In this setting, the LTL formula $\Eventually \mathsf{a}$, for example, will have an answer-indexed family that maps $\true$ to the set of traces in which $\mathsf{a}$ eventually holds, and $\false$ to the set of traces in which $\mathsf{a}$ never holds.

We define various operations on answer-indexed families, one for each kind of LTL constructor:
\begin{align*}
\Trueaif \ \true &= \itraces &(\Phi \Conjaif \Psi) \ \true &= \Phi \ \true \cap \Psi \ \true & (\Phi \Disjaif \Psi) \ \true &= \Phi \ \true \cup \Psi \ \true & (\Negaif \Phi) \ \true &= \Phi \ \false \\
\Trueaif \ \false &= \emptyset &(\Phi \Conjaif \Psi) \ \false &= \Phi \ \false \cup \Psi \ \false & (\Phi \Disjaif \Psi) \ \false &= \Phi \ \false \cap \Psi \ \false & (\Negaif \Phi) \ \false &= \Phi \ \true 
\end{align*}
Note that the set operations used for the $\false$ answer are always the duals of the operations used for the $\true$ answer, which means that for conventional LTL, the set produced for the $\false$ answer is always the complement of the set for the $\true$ answer. This means that the operator for negation ($\Negaif$) can simply swap the places of the set and its complement. This is akin to performing a conversion to negation normal form ``just-in-time'' as we evaluate a formula.

\noindent For atomic propositions $a$, the corresponding answer-indexed family maps $\true$ to the set of all traces that begin with a state containing $a$, and $\false$ to its complement:
\begin{align*}
\Propaif{\emph{a}} \ \true &= \{t \mid t \in \itraces \land \emph{a} \in t_0 \} \\
\Propaif{\emph{a}} \ \false &= \{t \mid t \in \itraces \land  \emph{a} \notin t_0\}
\end{align*}
The semantic operator for $\Next \varphi$ formulae prepends one state to all the corresponding traces for $\varphi$, analogously to the conventional LTL semantics in Figure~\ref{fig:ltlsemantics}:
\begin{align*}
(\Nextaif \Phi) \ \true &= \{t \mid  \drop{t}{1} \in \Phi \ \true \} \\
(\Nextaif \Phi) \ \false &= \{t \mid  \drop{t}{1} \in \Phi \ \false \}
\end{align*}
The $\true$ case of the semantic operator for $\varphi \Until \psi$ formulae is also defined analogously to Figure~\ref{fig:ltlsemantics}, with the $\false$ case being the complement:
\begin{align*}
(\Phi \Untilaif \Psi) \ \true &= \{t \mid \exists k. (\forall i < k. \drop{t}{i} \in \Phi \ \true) \wedge \drop{t}{k} \in \Psi \ \true \} \\
(\Phi \Untilaif \Psi) \ \false &= \{t \mid \forall k. (\exists i < k. \drop{t}{i} \in \Phi \ \false) \vee \drop{t}{k} \in \Psi \ \false \}
\end{align*}
Finally, we put all of these semantic operators to use in Figure~\ref{aifsemantics}, which gives a compositional, inductive semantics to conventional LTL using these operators.

\begin{figure}
\vspace{-2ex}
\centering
\begin{align*}
\semantics{\top} \ &=\ \Trueaif \\
\semantics{\emph{a}} \ &=\ \Propaif{\emph{a}} \\
\semantics{\neg \varphi} \ &=\ \Negaif \semantics{\varphi}\\
\semantics{\varphi \wedge \psi} \ &=\ \semantics{\varphi} \Conjaif \semantics{\psi}\\
\semantics{\varphi \vee \psi} \ &=\ \semantics{\varphi} \Disjaif \semantics{\psi} \\
\semantics{\Next \varphi} \ &=\ \Nextaif \semantics{\varphi} \\
\semantics{\varphi \Until \psi} \ &=\ \semantics{\varphi} \Untilaif \semantics{\psi}
\end{align*}

\caption{LTL semantics using answer-indexed families}
\label{aifsemantics}
\end{figure}

\begin{thm}[Equivalence to conventional semantics]\label{thm:ltlequiv}
Answer-indexed\linebreak[4] family LTL semantics assigns the same truth values to a given trace for a given formula as conventional LTL semantics:
\begin{itemize}
\item $(t \vDash \varphi) \Longleftrightarrow (t \in \semantics{\varphi}\ \true)$
\item $\neg (t \vDash \varphi) \Longleftrightarrow (t \in \semantics{\varphi}\ \false)$
\end{itemize}
\end{thm}
\begin{proof}
This is proven straightforwardly by induction on $\varphi$, justified in the same way as a conversion to negation normal form.
\end{proof}

\subsection{The Prepend Operation}
To give a semantics to \ltlthree, our answer-indexed families will produce definitive sets, rather than linear-time temporal properties. To this end, we will define an auxiliary operation on definitive sets called \emph{prepend}, written $\prepend X$, which gives all traces whose tails are in $X$:
$$\prepend X\; \triangleq \; \{ t \mid \drop{t}{1} \in X \}$$

\begin{thm}\label{thm:prependclosed}
The \emph{prepend} operation is closed for definitive sets. That is, if $X$ is definitive, then $\prepend X$ is definitive.
\end{thm}
\begin{proof}
We must show that $\lightning (\prepend X) = \prepend X$ for any definitive set $X$. Showing each direction separately:
\begin{description}
	\item[$\implies$] Given a definitive prefix $t \in \lightning(\prepend X)$, we must show that $t \in \prepend X$. If $t = \varepsilon$, then this implies that $\prepend X = \traces$ and therefore $t \in \prepend X$. If $t = \sigma{}u$, because $\prefixes (\prepend X) = \prepend (\prefixes X)$, we can conclude $\extensions \sigma{}u \subseteq \prepend (\downarrow X) $. Taking the tail of both sides, we can see that $\extensions u \subseteq {\prefixes X}$ and therefore $u \in X$ as $X$ is definitive. Prepending $\sigma$ to both sides, we conclude that $\sigma{}u \in \prepend X$ as required.
	\item[$\impliedby$] Given a prefix $t \in \prepend X$, we must show that $t \in \lightning (\prepend X)$. If $t = \varepsilon$, this means that $X = \prepend X = \lightning (\prepend X) = \traces$ as $X$ is definitive. If $t = \sigma{}u$, we know that $u \in X$. As $X$ is definitive, all extensions of $u$ are also in $X$. Therefore $\prepend (\extensions u) \subseteq X$ and thus $\sigma{}u \in \lightning (\prepend X)$.
\popQED
\end{description}
\end{proof}

\subsection[Semantics for LTL3]{Semantics for \ltlthree}
The semantic operators for \ltlthree resemble that of conventional LTL, except that now we work with definitive sets rather than linear-time temporal properties. 
\begin{align*}
\Trueaif _3 \ \true &= \traces &  \Trueaif _3 \ \false &= \emptyset\\
(\Negaif _3 \Phi) \ \true &= \Phi \ \false & (\Negaif _3 \Phi) \ \false &= \Phi \ \true\\
(\Phi \Conjaif_3 \Psi) \ \true &= \Phi \ \true \cap \Psi \ \true & (\Phi \Conjaif_3 \Psi) \ \false &= \Phi \ \false \dunion \Psi \ \false\\
(\Phi \Disjaif_3 \Psi) \ \true &= \Phi \ \true \dunion \Psi \ \true & (\Phi \Disjaif_3 \Psi) \ \false &= \Phi \ \false \cap \Psi \ \false \\
\end{align*}
All of the sets produced by these answer-indexed families are definitive, as $\traces$ and $\emptyset$ are both definitive sets and definitive sets are closed under intersection and definitive union. Unlike with conventional LTL, the set for the $\false$ answer is not the complement of the set for the $\true$ answer, as definitive sets are not closed under complement. The set for $\true$ contains all traces that are sufficient to definitively satisfy the formula, and the set for $\false$ contains all traces that are sufficient to definitively refute the formula.

For an atomic proposition $a$, the set for $\true$ contains all non-empty traces that begin with a state that satisfies $a$, and the set for $\false$ contains all non-empty traces that begin with a state that does not satisfy $a$. However, if $a$ is trivial, in the sense that \emph{all} or \emph{no} possible states satisfy $a$, then these sets are not definitive, as the excluded empty trace $\varepsilon$ would also be definitive for these sets. Thus, we take the definitive prefixes of these sets to account for this possibility:%
\begin{align*}
\Propaifthree{\emph{a}} \ \true &= \lightning\{t \mid t \neq \varepsilon \wedge  \emph{a} \in t_0 \} \\
\Propaifthree{\emph{a}} \ \false &= \lightning\{t \mid t \neq \varepsilon \wedge  \emph{a} \notin t_0\} \end{align*}
For the $\Next$ operator, we make use of the prepend operator, which by Theorem~\ref{thm:prependclosed} produces definitive sets:
\begin{align*}
(\Nextaif _3 \Phi) \ \true &= \prepend \ (\Phi \ \true) \\
(\Nextaif _3 \Phi) \ \false &= \prepend \ (\Phi \ \false)
\end{align*}
For the $\Until$ operator, we construct our semantics iteratively, building up by repeatedly prepending states. Here the notation $f^k$ indicates the self-composition of $f$ $k$ times, i.e.~$f^0(x) = x$ and $f^{k+1}(x) = f^k(f(x))$:
\begin{align*}
(\Phi \Untilaif _3 \Psi) \ \true &= \bigdunion_{k \in \mathbb{N}} f^k(\Psi \ \true), \ \text{where} \ f(X) = \prepend X \cap \Phi \ \true\\
(\Phi \Untilaif _3 \Psi) \ \false &= \bigcap_{k \in \mathbb{N}} f^k(\Psi \ \false), \ \text{where} \ f(X) = \prepend X \dunion \Phi \ \false
\end{align*}
Because definitive sets are closed under intersection, definitive union and the prepend operator, we can see that that our $\Untilaif$ operator also produces definitive sets by a simple inductive argument on the natural number $k$.  Using all of these operations, we construct an inductive, compositional semantics for \ltlthree  in Figure~\ref{ltl3semantics}.  
\begin{figure}
\vspace{-2ex}
\centering
\begin{align*}
\semantics{\top} _3 \ &=\ \Trueaif _3 \\
\semantics{\emph{a}} _3 \ &=\ \Propaifthree{\emph{a}} \\
\semantics{\neg \varphi} _3 \ &=\ \Negaif_3\ \semantics{\varphi} _3\\
\semantics{\varphi \wedge \psi} _3 \ &=\ \semantics{\varphi} _3  \Conjaif _3  \semantics{\psi} _3\\
\semantics{\varphi \vee \psi} _3 \ &=\ \semantics{\varphi} _3  \Disjaif _3  \semantics{\psi} _3 \\
\semantics{\Next \varphi} _3 \ &=\ \Nextaif_3\ \semantics{\varphi} _3 \\
\semantics{\varphi \Until \psi} _3 \ &=\ \semantics{\varphi}_3 \Untilaif_3 \semantics{\psi}_3 
\end{align*}

\caption{\ltlthree semantics using answer-indexed families}
\label{ltl3semantics}
\end{figure}
\begin{thm}[Equivalence to original \ltlthree definition]\label{thm:ltl3old}
	Let $t$ be a finite prefix and $\varphi$ be an LTL formula. Then:
\begin{itemize}
	\item $t \in \semantics{\varphi}_3\ \true \iff \forall u \in \itraces.\ tu \in \semantics{\varphi}\ \true$
	\item $t \in \semantics{\varphi}_3\ \false \iff \forall u \in \itraces.\ tu \in \semantics{\varphi}\ \false$
\end{itemize}
\end{thm}
\begin{proof}
	This follows directly from the definition of definitive sets, as $\semantics{\varphi}_3\ \true$ and $\semantics{\varphi}_3\ \false$ are both definitive.
\end{proof}
\noindent Theorem~\ref{thm:ltl3old} shows that our inductive semantics coincides with the original non-inductive semantics given for \ltlthree. If we view our semantics through the lens of the isomorphism in Theorem~\ref{thm:bijection}, however, we see that this semantics is also equivalent to the semantics of conventional LTL:
\begin{thm}[Equivalence to conventional LTL]\label{thm:ltl3equiv}
For all formulae $\varphi$:
\begin{itemize}
\item $\Prop (\semantics{\varphi} _3\ \true) = \semantics{\varphi}\ \true$
\item $\Prop (\semantics{\varphi} _3\ \false) = \semantics{\varphi}\ \false$
\end{itemize}
\end{thm}
\begin{proof}
The two statements are shown simultaneously by induction on $\varphi$:
\begin{itemize}[leftmargin=7em]
	\item[$\varphi = \top$:]	$\Prop(\traces) = \itraces$ by definition.
	\item[$\varphi = a$:]  Because $\Prop(\lightning S) = \Prop(S)$ as seen in the proof of Theorem~\ref{thm:bijection}, $\Prop(\Propaifthree{a} \true) = \Propaif{a}\ \true$ and likewise $\Prop(\Propaifthree{a} \false) = \Propaif{a}\ \false$. 
	\item[$\varphi = \neg \varphi'$:] Follows from inductive hypotheses. 
	\item[$\varphi = \varphi' \land \psi'$:] Follows from inductive hypotheses as $\Prop$ preserves greatest lower  and least upper bounds.
	\item[$\varphi = \Next \varphi'$:] Follows from inductive hypotheses as the prepend operator $\prepend$ commutes with $\Prop$.
	\item[$\varphi = \varphi' \Until \psi'$:] Because $\Prop$ commutes with $\prepend$ and preserves least upper and greatest lower bounds, we can show that $\Prop(\bigdunion_{k \in \mathbb{N}} f^k(\semantics{\psi'}_3 \, \true))$ where $f(X) = \prepend X \cap \semantics{\varphi}_3 \ \true$, is equal to $\bigcup_{k \in \mathbb{N}} g^k(\Prop(\semantics{\psi'}_3 \ \true))$ where  $g(X) = \prepend X \cap \Prop(\semantics{\varphi'}_3 \ \true)$ by induction on on the natural number $k$. By the inductive hypotheses, this is equal to $\bigcup_{k \in \mathbb{N}} g^k(\semantics{\psi'} \ \true)$ where  $g(X) = \prepend X \cap \semantics{\varphi'} \ \true$.   This can be shown by another simple induction to be equal to the original definition in the conventional LTL semantics $\{t \mid \exists k. (\forall i < k. \drop{t}{i} \in \semantics{\varphi'} \ \true) \wedge \drop{t}{k} \in \semantics{\psi'} \ \true \}$. The cases for the $\false$ answer are proved similarly.
\popQED
\end{itemize}
\end{proof}
\noindent Because of this equivalence theorem, we can now express the relationship between the set for $\true$ and the set for $\false$ in our \ltlthree semantics. In \ltlthree, while the two sets do not overlap, they are not perfect complements of each other as they were in conventional LTL, as definitive sets are not closed under complement. Instead, the $\false$ set is the definitive set corresponding to the \emph{linear-time temporal property} containing all infinite traces not in the $\true$ set. 
\begin{thm}[Excluded Middle]\label{thm:em}
For all formulae $\varphi$: $$ \semantics{\varphi}_3\ \true = \lightning(\itraces \setminus \semantics{\varphi}_3\ \false) \qquad \qquad \text{and} \qquad \qquad \semantics{\varphi}_3\ \false = \lightning(\itraces \setminus \semantics{\varphi}_3\ \true) $$	
\end{thm}
\begin{proof} It is a straightforward consequence of Theorem~\ref{thm:ltlequiv} that $\semantics{\varphi}\ \true = \itraces \setminus \semantics{\varphi}\ \false$ (*). Then:
\begin{displaymath}
	\begin{array}{lclr}
		 \semantics{\varphi}_3\ \true & =& \Definitives(\Prop ( \semantics{\varphi}_3\ \true)) & \text{(Theorem~\ref{thm:bijection})}\\
         & =& \Definitives(\semantics{\varphi}\ \true) & \text{(Theorem~\ref{thm:ltl3equiv})}\\
         & =& \Definitives(\itraces \setminus \semantics{\varphi}\ \false) & \text{(*)}\\	 
         & =& \Definitives(\itraces \setminus \Prop (\semantics{\varphi}_3\ \false)) & \text{(Theorem~\ref{thm:ltl3equiv})}\\
         & =& \Definitives(\itraces \setminus (\semantics{\varphi}_3\ \false \cap \itraces)) &\qquad \text{(Definition of $\Prop$)}\\         
         & =& \lightning (\itraces \setminus \semantics{\varphi}_3\ \false) & \text{ (Definition of $\Definitives$)}\\[-4ex]
	\end{array}
\end{displaymath}	
\end{proof}

\section{Formula Progression}
\label{sec:fp}

\emph{Formula progression} is a technique first introduced by Kabanza et al.~\cite{progress1,progress2} that evaluates a formula stepwise against states in a style reminiscent of operational semantics or the Brzozowski derivative. This technique was used by O'Connor {\&} Wickstr\"om~\cite{quickstrom} as the basis for their testing algorithm, and by Bauer {\&} Falcone~\cite{fpltl3fm} for decentralised monitoring of component-based systems.  Figure~\ref{fig:fp} gives an overview of formula progression rules for LTL\@. The judgement \plat{$\varphi \xrightarrow{\sigma} \psi$} states that, to prove $\varphi$, it suffices to prove $\psi$ for the tail of our trace if the head of our trace is $\sigma \in \Sigma$. Note that these rules are total and syntax-directed on the left-hand formula $\varphi$. This means that these rules taken together constitute a definition of a total function that takes $\varphi$ and $\sigma$ as input and produces $\psi$ as output. We generalise this notation to finite prefixes, so that for a finite trace $t = \sigma_0\dots\sigma_n$, the notation $\varphi \xrightarrow{t} \psi$ just means $\varphi \xrightarrow{\sigma_0} \cdots \xrightarrow{\sigma_n} \psi$. Repeated application of these rules, however, can lead to exponential blowup in the size of the formula. While both O'Connor {\&} Wickstr\"om~\cite{quickstrom} and Bauer {\&} Falcone~\cite{fpltl3fm} report that interleaving this progression with formula simplification at each step keeps the formulae tractable for most practical use cases, Ro\c{s}u {\&} Havelund \cite{rosu} warn that pathological exponential cases still exist.

\begin{figure}
\vspace{-2ex}
\begin{gather*}
\boxed{\varphi \xrightarrow{\sigma} \psi}\\
\inferrule{ }{\top \xrightarrow{\sigma} \top} \qquad 
\inferrule{a \in \sigma}{a \xrightarrow{\sigma} \top} \qquad 
\inferrule{a \notin \sigma}{a \xrightarrow{\sigma} \bot} \qquad 
 \inferrule{ }{\Next \varphi \xrightarrow{\sigma} \varphi} \qquad
\inferrule{ \varphi \xrightarrow{\sigma} \varphi' \\
              \psi \xrightarrow{\sigma} \psi' }{\varphi \lor \psi \xrightarrow{\sigma} \varphi' \lor \psi'} \\ 
 \inferrule{ \varphi \xrightarrow{\sigma} \varphi' }{\neg \varphi \xrightarrow{\sigma} \neg \varphi'}\qquad 
\inferrule{ \varphi \xrightarrow{\sigma} \varphi' \\
              \psi \xrightarrow{\sigma} \psi' }{\varphi \land \psi \xrightarrow{\sigma} \varphi' \land \psi'} \qquad 
\inferrule{ \psi \xrightarrow{\sigma} \psi' \\ \varphi \xrightarrow{\sigma} \varphi'}{\varphi \Until \psi \xrightarrow{\sigma} \psi' \lor (\varphi' \land (\varphi \Until \psi)) }
\end{gather*}
\caption{Rules for formula progression}
\label{fig:fp}	
\end{figure}

Bauer {\&} Falcone~\cite{fpltl3fm} state that formula progression can serve as an \emph{alternative semantics} for \ltlthree on finite traces, where a formula $\varphi$ is considered definitively true for a finite trace $t$ iff $\varphi \xrightarrow{t} \top$, definitively false iff $\varphi \xrightarrow{t} \bot$, and is unknown otherwise. While it goes unmentioned in their paper, here the implicit simplification steps are not just a performance optimisation, but are vital to ensure that the semantics given via formula progression is complete with respect to the standard \ltlthree semantics. To see why, consider the formula $\Eventually a$. Let $\sigma_a$ be a state where $a \in \sigma_a$. Then the formula $\Eventually a$ should be considered definitively true for the trace consisting of just $\sigma_a$. The formula generated by our formula progression rules, however, would be $\top \lor (\top \land \Eventually a)$, which yields the desired formula $\top$ only after logical simplifications are applied. While in this case, the simplifications required are just identities of propositional logic, in general such straightforward simplifications alone are insufficient. For example, consider the formula $(\Next a) \lor (\Eventually \neg a)$. According to the semantics of \ltlthree presented above, this formula should be considered definitively true for the empty trace $\varepsilon$, as it is a tautology. Temporally local simplifications such as those used by O'Connor {\&} Wickstr\"om~\cite{quickstrom}, however, would not be able to determine that this formula is a tautology until after one state has been observed. Therefore, in order for formula progression to align correctly with the semantics of \ltlthree, the simplification must transform \emph{all} tautologies into $\top$ and \emph{all} absurdities into $\bot$. A simple, although slow way to implement such a simplifier would be to convert both the formula and its negation into B\"uchi automata, and perform cycle detection to check for emptiness. For our development, we abstract away from such syntactic simplification procedures by working only on the level of our model-based semantics. As can be seen in our Theorem~\ref{thm:fp3} given below, we do not seek a specific syntactic tautology $\top$ or absurdity $\bot$, but rather refer to any formula with trivial semantics. A purely syntactic characterisation, by contrast, would require a full accounting of the simplification procedure, which is outside the scope of our development here.

The rules given in Figure~\ref{fig:fp} operate on one state at a time, whereas our semantics are on the level of entire traces. Therefore, in order to show soundness and completeness (for finite traces) of our formula progression rules with respect to our semantics, we must first prove two lemmas which relate a single step of formula progression to our semantics.

The first lemma states that for one step of formula progression $\varphi \xrightarrow{\sigma} \varphi'$, prepending $\sigma$ to the traces that satisfy/refute the \emph{output} formula $\varphi'$ yields traces that satisfy (resp.\ refute) the \emph{input} formula $\varphi$.
  
 \begin{lem}\label{thm:fp1} Let $\varphi$ and $\varphi'$ be formulae and $\sigma$ be a state such that $\varphi \xrightarrow{\sigma} \varphi'$. Then:
 \begin{itemize}\item
 	$ \prepend (\semantics{\varphi'}_3\ \true) \cap \{ t \mid t_0 = \sigma \} \subseteq \semantics{\varphi}_3\ \true  $
 	\item 
 	 	$ \prepend (\semantics{\varphi'}_3\ \false) \cap \{ t \mid t_0 = \sigma \} \subseteq \semantics{\varphi}_3\ \false  $
 \end{itemize}
 \end{lem}
\begin{proof}
	The two statements are shown simultaneously by structural induction on the formula $\varphi$ (which, as our rules are syntax directed, uniquely determines the output formula $\varphi'$). The base cases for $\varphi = \top$ and $\varphi = a$ as well as the inductive cases for the next operator $\Next$ follow directly from definitions. Of the other inductive cases, the cases for conjunction and disjunction require the use of the distributive properties of the lattice of definitive sets, as well as the fact that the prepend operator $\prepend$ distributes over intersection and definitive union. The cases for negation follow directly from the inductive hypotheses, whereas the cases for the until operator $\Until$ require unfolding of the big unions and intersections in the definition of the semantic operator $\Untilaif\!_3$ by one step.
\end{proof}
\noindent The second lemma states that those traces that satisfy the \emph{input} formula $\varphi$ and begin with the state $\sigma$ will have tails that satisfy the \emph{output} formula $\varphi'$.
\begin{lem}\label{thm:fp2}  Let $\varphi$ and $\varphi'$ be formulae and $\sigma$ be a state such that $\varphi \xrightarrow{\sigma} \varphi'$. Then:
 \begin{itemize}\item
 	$ \semantics{\varphi}_3\ \true\ \cap \{ t \mid t_0 = \sigma \} \subseteq \prepend (\semantics{\varphi'}_3\ \true)  $
 	\item 
 	 	$ \semantics{\varphi}_3\ \false\ \cap \{ t \mid t_0 = \sigma \} \subseteq \prepend (\semantics{\varphi'}_3\ \false)  $
 \end{itemize}
\end{lem}
\begin{proof}
As with Theorem~\ref{thm:fp1}, the two statements are shown simultaneously by structural induction on the formula $\varphi$. The base cases and the cases for the next operator $\Next$ are shown just by unfolding definitions, the cases for conjunction and disjunction are shown by use of distributive properties including those of the prepend operator $\prepend$, negation proceeds directly from the induction hypotheses, and the until operator $\Until$ requires unfolding of the semantic operator $\Untilaif\!_3$ by one step.
\end{proof}
\noindent By combining these two lemmas, we can inductively prove a theorem that relates formula progression to our semantics on the level of entire finite traces. This resembles the informal definition of formula progression semantics given by Bauer {\&} Falcone~\cite{fpltl3fm}, but with the syntactic requirement that the ultimate formula be $\top$ or $\bot$ replaced by a semantic requirement that it has trivial semantics.
\begin{thm}\label{thm:fp3} Let $t \in \ftraces$ be a finite trace. Then, for all formulae $\varphi$ and $\varphi'$ where $\varphi \xrightarrow{t} \varphi'$:
\begin{itemize}
\item  $t \in \semantics{\varphi}_3\ \true$ if and only if $\semantics{\varphi'}_3\ \true = \traces$.
\item  $t \in \semantics{\varphi}_3\ \false$ if and only if $\semantics{\varphi'}_3\ \false = \traces$.	
\end{itemize}	
\end{thm}
\begin{proof}
By induction on the length of the trace $t$ (where $\varphi$ and $\varphi'$ are kept arbitrary). The second statement for $\false$ is proved identically to the first for $\true$, so we present the proof only for $\true$ here.
	\begin{description}
	\item[Base Case $(t = \varepsilon)$] It suffices to show that $\varepsilon \in \semantics{\varphi}_3\ \true$ iff $\semantics{\varphi}_3\ \true = \traces$. Because $\semantics{\varphi}_3\ \true$ is a definitive set, and any extension of a definitive prefix is also a definitive prefix, as $\varepsilon$ is in $\semantics{\varphi}_3\ \true$, we can conclude that all traces (i.e.\ extensions of $\varepsilon$) are in $\semantics{\varphi}_3\ \true$. The reverse direction of the iff is straightforward.
	\item[Inductive Case $(t = \sigma u)$] We know that $\varphi_0 \xrightarrow{\sigma} \varphi \xrightarrow{u} \varphi'$ and have the inductive hypothesis that $u \in \semantics{\varphi}_3\ \true\linebreak[2] \iff \semantics{\varphi'}_3\ \true = \traces$. We must show that $\sigma u \in \semantics{\varphi_0}_3\ \true \iff \semantics{\varphi'}_3\ \true = \traces$. Therefore, by the inductive hypothesis, it suffices to show $\sigma u \in \semantics{\varphi_0}_3\ \true \iff u \in \semantics{\varphi}_3\ \true$. Showing each direction separately:
	\begin{description}
	\item[$\implies$] By Theorem~\ref{thm:fp2} we can conclude that $\sigma u \in \prepend (\semantics{\phi}_3\ \true)$ and thus that $u \in \semantics{\varphi}_3\ \true$ by the definition of the prepend operator $\prepend$. 
	\item[$\impliedby$] By the definition of the prepend operator $\prepend$ we can conclude that $\sigma u \in \prepend(\semantics{\varphi}_3\ \true)$ and thus that $\sigma u \in \semantics{\varphi_0}_3\ \true$ by Theorem~\ref{thm:fp1}. 
\popQED
	\end{description}
\end{description}	
\end{proof}
\noindent Theorem~\ref{thm:fp3} is both a \emph{soundness} and \emph{completeness} proof for formula progression semantics with respect to our model-based semantics, up to finite traces. Soundness here means that a formula $\varphi$ will only evaluate in formula progression to a tautology for a trace $t$ when $t$ is in $\semantics{\varphi}_3\ \true$, and likewise will only evaluate to an absurdity when $t$ is in $\semantics{\varphi}_3\ \false$. This is the $\impliedby\!\!$ direction of the iff in Theorem~\ref{thm:fp3}. Completeness (or adequacy) up to finite traces means that all finite prefixes that definitively confirm the formula will evaluate in formula progression to a tautology, and all finite prefixes that definitively refute the formula will evaluate to an absurdity. This is the $\implies\!\!$ direction of the iff in Theorem~\ref{thm:fp3}.

The progressed formula evaluating to a tautology or an absurdity should be understood purely as a semantic description of when progression has reached a formula with trivial semantics, rather than as an implementation prescription for an explicit universality or emptiness check. Any practical implementation must therefore balance completeness with respect to this characterisation against efficiency. This trade-off can be addressed either by employing complete but potentially expensive procedures, or by adopting tractable but incomplete syntactic simplification techniques~\cite{quickstrom}.

\section{Properties and Monitorability}\label{sec:properties}
\subsection{Linear-time Temporal Properties}
Linear-time temporal properties can be broadly categorised into \emph{safety} properties, which state that something ``bad'' does not happen during execution, and \emph{liveness} properties, which state that something ``good'' eventually occurs during execution~\cite{Lam77}. Alpern {\&} Schneider~\cite{alpernschneider} formalise Lamport's informal concepts of safety and liveness properties by defining a safety property as one that can be definitely refuted by a finite prefix of a trace, and a liveness property $P$ as one such that any finite trace can be extended to a member of $P$. They define a topological space where safety properties are those sets that are (limit-)closed and liveness properties are those sets that are dense. We also see in later work~\cite{kupfermanvardi,quant} the concept of \emph{guarantee} (or \emph{co-safety}) properties and \emph{morbidity}  (or \emph{co-liveness}) properties, the complements of safety and liveness properties respectively. A \emph{guarantee} property can always be definitively \emph{confirmed} by a finite prefix of a trace,\footnote{Thus, \emph{guarantee} could be seen as an alternative formalisation of Lamport's informal concept of a liveness property~\cite{Lam77}, different from the standard formalisation of Alpern {\&} Schneider~\cite{alpernschneider}. This formalisation of liveness is employed, e.g., by van Glabbeek~\cite{vG10}.} whereas any finite prefix can be extended in such a way as to \emph{refute} a given \emph{morbidity} property.

Our definitive sets include those finite prefixes that can confirm (or refute) the property, enabling us to express these insights about finite prefixes directly. We shall use this to define a new operator on properties, which we call the \emph{guarantee kernel}, which, when viewed as an interior operator, gives rise to an equivalent topology to that of Alpern {\&} Schneider~\cite{alpernschneider}, enabling us to re-prove their classic result---that every property is the intersection of a safety and liveness property---by other means. These definitions also enable precise definitions of \emph{monitorability} classes, which we discuss in Section~\ref{sec:monitorability}.

\subsubsection{Guarantee and Safety Properties}

\begin{definition}[Guarantee kernel]
	The \emph{guarantee kernel} of a property $P$, written $\guar{P}$, is the set of all infinite traces that are extensions of finite definitive prefixes of $P$:
	$$
	\guar{P} \;\triangleq\; \extensions(\dclose{}\!P \cap \ftraces) \cap \itraces
	$$
	Intuitively, traces in the guarantee kernel $\guar{P}$ are those traces $t \in P$ for which membership in $P$ can be confirmed after observing merely a finite prefix of $t$. This means that $\guar{P}$ is the largest subset of $P$ that is a guarantee property. Thus, we can characterise guarantee properties using the guarantee kernel:
	$$
	P\ \text{is a guarantee property}\;\; \longleftrightarrow\;\; \guar{P} = P
	$$
	The guarantee kernel is a \emph{kernel operator}, that is, it is monotone ($P \subseteq R$ implies $\guar{P} \subseteq \guar{R}$), it is idempotent ($\guar{\guar{P}} = \guar{P}$), and it is a subset of the
	original property ($\guar{P} \subseteq P$). 
	
	\begin{thm}\label{thm:gkint}
		The guarantee kernel distributes over binary intersection, i.e., $\guar{P} \cap \guar{R} = \guar{P \cap R}$.
	\end{thm}
	\begin{proof} Showing each direction separately:
	\begin{description}
	\item[$\guar{P} \cap \guar{R} \subseteq \guar{P \cap R}$]: Let $u$ be an infinite trace in $\guar{P} \cap \guar{R}$. Then there must exist a definitive finite prefix for $P$, $t_P \in \prefixes{u}$. Similarly there must exist a definitive finite prefix for $R$, $t_R \in \prefixes{u}$. Because both $t_P$ and $t_R$ are prefixes of $u$, either $t_P \in \prefixes{t_R}$ or $t_R \in \prefixes{t_P}$. If $t_P\in \prefixes{t_R}$, then $t_R$ must also be a definitive prefix for $P$ as definitive sets are closed under extension, so $t_R$ is definitive finite prefix for $P \cap R$ and thus $u \in \guar{P \cap R}$. Likewise for $t_P$ and $R$ when $t_R \in \prefixes{t_P}$.
		\item[$\guar{P} \cap \guar{R} \supseteq \guar{P \cap R}$]: Straightforward as the guarantee kernel is monotone.\qedhere
	\end{description}
	\end{proof}
	
	\begin{counterexample}[Guarantee kernel does not distribute over union]
 		Consider the properties $\Always \Eventually a$ and $\Eventually \Always \neg a$. The guarantee kernel of each is $\emptyset$, but their union is trivially true, thus having a guarantee kernel of $\Sigma^\omega$.
	\end{counterexample}
\end{definition}

\noindent As a safety property is just the complement of a guarantee property,
we can characterise safety properties analogously\footnote{Throughout this paper, the complement notation is only ever applied to \emph{properties}, not definitive sets. So $\compl{P} \triangleq \itraces \setminus P$.}: $$
	P\ \text{is a safety property}\;\; \longleftrightarrow\;\; \guar{\compl{P}} = \compl{P}
	$$
\begin{definition}[Safety closure]
	The \emph{safety closure} of a property $P$, written $\safety{P}$, is the set of all infinite traces which cannot be shown \emph{not} to be in $P$ by examining only a finite prefix:
	$$
		\safety{P} \; \triangleq \; \compl{(\guar{\compl{P}})}
	$$
	It follows that $\safety{P}$ is the smallest superset of $P$ that is a safety property. Thus, we can restate our above characterisation of safety properties as:
	$$
	P\ \text{is a safety property}\;\; \longleftrightarrow\;\; \safety{P} = P
	$$
	As the dual of the guarantee kernel, the safety closure is a \emph{closure operator}. That is, it is monotone ($P \subseteq R$ implies $\safety{P} \subseteq \safety{R}$), it is idempotent ($\safety{\safety{P}} = \safety{P}$), and it is a superset of the
	original property ($\safety{P} \supseteq P$). This operator coincides with the closure operator used to define the topology in Alpern {\&} Schneider~\cite{alpernschneider}. Thus, as with Alpern {\&} Schneider, safety properties are closed sets in a topology, and therefore guarantee properties are the open sets, and the guarantee kernel is the \emph{interior operator} for the topology.
\end{definition}
\begin{thm}[The space of properties is metrisable]
Let $P$ be a property, with the metric:
$$
d(t,u) = \begin{cases} 0 & \text{if}\ t = u \\ 2^{-\sup\{\ell\, \mid\, \forall i < \ell.\,  t_i = u_i \}} & \text{if}\ t \neq u \end{cases}
$$
Then $P$ is a guarantee property iff $P$ is open in the metric topology, i.e.\ $\forall t \in P.\ \exists \varepsilon > 0.\ \{ u \mid d(t,u) < \varepsilon \} \subseteq P $.
\end{thm}
\begin{proof}
\begin{description}
\item[$\implies$] Assume $P$ is a guarantee property. Let $t \in P$. Then, because $P$ is guarantee, there must exist some finite $p \in \prefixes{t}$ that is definitive for $P$, i.e.\ $\extensions{p} \cap \itraces \subseteq P$. Set $r = |p| - 1$ and $\varepsilon = 2^{-r}$. Then the metric ball of radius $\varepsilon$, i.e.\ $\{ u \mid d(t,u) < \varepsilon \}$, is the set of all infinite extensions of $p$, i.e $\extensions{p} \cap \itraces$, because $d(t,u) < 2^{-r}$ iff $t$ and $u$ agree for a prefix of length $r + 1$ --- that is, $p$. Since $t$ was arbitrary, this shows that $\forall t \in P.\ \exists \varepsilon > 0.\ \{ u \mid d(t,u) < \varepsilon \} \subseteq P$ as required.
\item[$\impliedby$] Assume $P$ is open in the metric topology, i.e.\ that $\forall t \in P.\ \exists \varepsilon > 0.\ \{ u \mid d(t,u) < \varepsilon \} \subseteq P $.  We shall show that $P$ is a guarantee property, by showing that $P \subseteq \guar{P}$. Assume $t \in P$. Let $\varepsilon > 0$ be such that  $\{ u \mid d(t,u) < \varepsilon \} \subseteq P $. By the Archimedean property, there must be some natural number $r$ such that $2^{-r} < \varepsilon$. The ball of radius $2^{-r}$, i.e.\ $\{ u \mid d(t,u) < 2^{-r} \}$ is therefore contained within the ball of radius $\varepsilon$, which, by our openness assumption, must in turn be contained within $P$. Because $d(t,u) < 2^{-r}$ iff $t$ and $u$ agree for a prefix of length $r + 1$, let $p \in \prefixes{t}$ be a prefix of $t$ of length $r + 1$. Then the ball of radius $2^{-r}$ is exactly $\extensions{p} \cap \itraces$. Therefore, $\extensions{p} \cap \itraces \subseteq \{ u \mid d(t,u) < \varepsilon \} \subseteq P$, i.e.\ $p \in \lightning{P}$. This shows, as $t$ was arbitrary, that all $t \in P$ are the extension of some finite definitive $p \in \lightning{P}$, and therefore that $t \in \guar{P}$. Therefore $P \subseteq \guar{P}$ and $P$ is a guarantee property.\qedhere 
\end{description}	
\end{proof}
As a consequence, the collection of guarantee properties over $\Sigma$ constitutes a topology: it contains $\emptyset$ and $\Sigma^\omega$, and is closed under arbitrary unions and finite intersections.

\subsubsection{Liveness and Morbidity}
Liveness properties are those that can never be definitively refuted by a finite prefix. Thus a property $P$ represents a liveness property iff all finite prefixes are prefixes of traces in $P$, i.e., $\ftraces \subseteq \prefixes P$. Equivalently, liveness properties are \emph{dense}: $P$ is a liveness property iff $\safety{P} = \itraces$. 

As morbidity properties are the complement of liveness properties, we can say $P$ represents a morbidity property iff all finite prefixes are prefixes of traces in its complement, i.e., $\ftraces \subseteq \prefixes (\compl{P})$. Equivalently, morbidity properties have an empty interior: $P$ is a morbidity property iff $\guar{P} = \emptyset$.

\begin{thm}[Alpern {\&} Schneider, redux]
	Every property is the union of a guarantee and a morbidity property.
\end{thm}
\begin{proof}
For a given property $P$, we can decompose it straightforwardly as follows:
$$
P = \underbrace{\;\;\guar{P}\;\;}_\text{guarantee} \;\;\cup\;\; \underbrace{(P \setminus \guar{P})}_\text{morbidity}
$$ 
$\guar{P}$ is clearly a guarantee property. It remains to show that $P \setminus \guar{P}$ is a morbidity property. We shall do this by showing that its interior is empty, i.e.:
$$
\begin{array}{lcll}
	\guar{P \setminus \guar{P}} &=& \guar{P \cap \compl{(\guar{P})}}\\[0.3em]
	& = &  \guar{P} \cap \guar{\compl{(\guar{P})}} &\quad \text{(Theorem~\ref{thm:gkint})}\\[0.3em]
	& \subseteq & \guar{P} \cap \compl{(\guar{P})} & \quad\text{(g.k.~is a subset)}\\[0.1em]
	& = & \emptyset  \\[-4ex]
\end{array}
$$
\end{proof}
\noindent As we have proved that every property is the union of a guarantee and a morbidity property, we can then show that every property is the intersection of a safety and liveness property, the famous result of Alpern {\&} Schneider~\cite{alpernschneider}, simply by taking the complement.

\subsection{Monitorability}
\label{sec:monitorability}
There are some properties for which \ltlthree will \emph{always} give a \unknown\@ answer. For example, take the standard \emph{request/acknowledge} format:
$$
\varphi = \Always (r \Rightarrow \Eventually a)
$$%
which states that all requests ($r$) must eventually be acknowledged ($a$). For every finite prefix $u$, we have $ur^\omega \in \semantics{\varphi}_3\ \false $ and $ua^\omega \in \semantics{\varphi}_3\ \true$. As the $\false$ and $\true$ answers are non-overlapping (Theorem~\ref{thm:em}), $u$ cannot be a definitive prefix.  Therefore, \emph{all} finite prefixes are not definitive, so no run-time monitor will ever be useful for this property.
In order for run-time monitoring to be useful, temporal properties must be \emph{monitorable}. Exactly what \emph{monitorable} means, however, is a murky question, and several possible answers exist in the literature. We formally characterise and compare several of these definitions below. Our findings are summarised in Figure~\ref{fig:moncube}. In this diagram, an arrow indicates a \emph{strict} inequality. For example, all safety properties are bad monitorable (see Theorem~\ref{thm:safetybad}), but not all bad monitorable properties are safety properties. 
\begin{figure*}
\begin{center}
\begin{tikzpicture}
\node at (0,0) {$\mathsf{Strong}$};
\node at (0,-6.5) {$\mathsf{Weak}$};
\node at (0,-0.5) (sm) {$\cap$};
\node at (-1.5,-1.5) (saf) {$\mathsf{Safety}$};
\node at (1.5,-1.5) (guar) {$\mathsf{Guarantee}$};
\draw[thick] (sm) -- (saf) (sm) -- (guar);
\node at (-1.5,-4.75) (bad) {$\mathsf{Bad}$};
\node at (1.5,-4.75) (good) {$\mathsf{Good}$};
\draw[thick,<-] (bad) -- (saf);
\draw[thick,<-] (good) -- (guar);
\node at (0,-6) (wm) {$\cup$};
\node at (0,-2.5) (usg) {$\cup$};
\node at (0,-3.5) (obl) {$\mathsf{Obligation}$};
\node at (0,-4.5) (bm) {$\mathsf{BLS}$};
\draw[thick] (bad) -- (wm) (good) -- (wm);
\draw[thick,] (saf) -- (usg);
\draw[thick,] (guar) -- (usg);
\draw[thick,->] (usg) -- (obl);
\draw[thick,->] (obl) -- (bm);
\draw[thick,->] (bm) -- (wm);
\end{tikzpicture}
\end{center}
\caption{The monitorability lattice}
\label{fig:moncube}
\end{figure*}
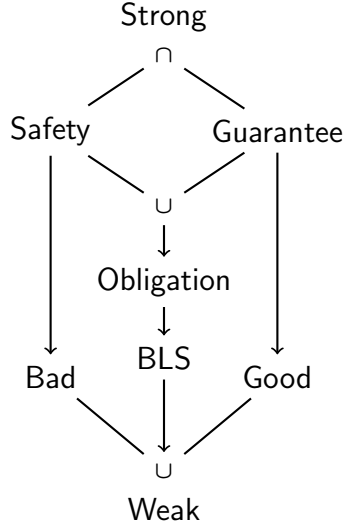

\subsubsection{Good and bad monitorability}\label{sec:goodbad}
Kupferman {\&} Vardi~\cite{kupfermanvardi} define the \emph{bad prefixes} of a property $P \subseteq \itraces$ as those prefixes that cannot be extended to a trace that is in $P$, and further define \emph{good prefixes} as those for whom all infinite extensions are in $P$. In our notation, the good prefixes of $P$ are exactly the definitive prefixes $\dclose P$, and the bad prefixes are the definitive prefixes of the complement $\dclose \compl{P}$. From this, we define a property $P$ as \emph{good monitorable} iff there exists a good finite prefix of $P$, \emph{or} there are no traces in $P$. Similarly, a property $P$ is \emph{bad monitorable} iff there exists a bad finite prefix of $P$, \emph{or} if every trace is in $P$.
\begin{definition}[Good and bad monitorability]
	$$P\ \text{is good monitorable} \; \longleftrightarrow \; \dclose{\!P} \cap \ftraces \neq \emptyset \lor P = \emptyset \; \longleftrightarrow  \;\guar{P} \neq \emptyset \lor P = \emptyset $$
		$$P\ \text{is bad monitorable} \; \longleftrightarrow \; \dclose{}\!\compl{P} \cap \ftraces \neq \emptyset \lor P = \itraces  \; \longleftrightarrow  \;\safety{P} \neq \itraces \lor P = \itraces $$
\end{definition}
\begin{thm}\label{thm:safetybad}
All safety properties are bad monitorable, and all guarantee properties are good monitorable.
\end{thm}
\begin{proof}
Consider a guarantee property $P$. If $P = \emptyset$ then it is trivially good monitorable. Otherwise, let $t$ be an infinite trace in $P$. Because $P$ is a guarantee property, $\guar{P} = P$, and thus $t$ must be the extension of some finite good prefix $p \in \dclose P$. Thus $P$ is good monitorable. For safety properties and bad monitorability, we take the complement.
\end{proof}
\subsubsection{Weak monitorability}
Next we define the largest class of properties for which monitors can be of any use, \emph{weak monitorability}. A property is weak monitorable if there exists at least one definitive prefix, good or bad. This definition of monitorability is often attributed to Pnueli {\&} Zaks~\cite{PnueliZaks}, and is also called ``weak monitorability'' by Chen et al.~\cite{ChenWM}.
\begin{definition}{Weak Monitorability}
$$P\ \text{is weak monitorable}\; \longleftrightarrow \;   \guar{P} \cup \guar{\compl{P}} \neq \emptyset \; \longleftrightarrow \; \safety{P} \setminus \guar{P} \neq \itraces $$
We call the set $\safety{P} \setminus \guar{P}$ the \emph{unmonitorable traces}  (topologically speaking, the \emph{frontier}) of $P$, because these are precisely those infinite traces for which definitive answers cannot be given after observing only a finite prefix.

As can be seen in Figure~\ref{fig:moncube}, the class of weak monitorable properties is just the union of the classes of good monitorable and bad monitorable properties.

Weak monitorability is closed under negation, but not union or intersection. 
For example, take properties $P = \Always \Eventually b \lor \Next a$ and $Q = \Always \Eventually b \lor \Next (\neg a)$. Both $P$ and $Q$ are weak monitorable, but $P \land Q$ is not monitorable.
\end{definition}
\begin{figure*}
\begin{center}
\begin{tikzpicture}
	\draw[dotted,thick, fill=orange!20!white] (0,5) -- (-1.5,5.5) -- (0,6);
	\draw[draw=orange!70!black] (0,0.5) -- (0,6);
	\draw[thick,dotted,fill=blue!10!white] (0,0.5) -- (-1.5,1) -- (0,1.5) -- (-2,2.25) -- (0,3);
	\draw[thick,dotted] (0,3.25) -- (-2,4) -- (0,4.75);
	\draw[very thick, draw=blue!70!black] (0,0.5) -- (0,4);
	\draw[thick] (-0.1,6) -- (0.1,6);
	\draw[thick] (-0.1,5) -- (0.1,5);
	\draw[thick] (-0.1,3) -- (0.1,3);
	\draw[thick] (-0.1,0.5) -- (0.1,0.5);
	\draw[thick] (-0.1,4) -- (0.1,4);
	\node at (-0.5,2.25) {\scriptsize $\dclose{P}$};
	\node at (-0.5,1) {\scriptsize $\dclose{P}$};
	\node at (-0.5,5.5) {\scriptsize $\dclose{\compl{P}}$};
	\node at (-0.75,4) {\scriptsize $\dclose{(\safety{P} \setminus \guar{P})}$};
	\draw[thick] (0.9,0.5) -- (1.1,0.5);
	\draw[thick] (0.9,4) -- (1.1,4);
	\draw[thick] (1,0.5) to node[auto,swap] {$P$} (1,4) ;
	\draw[thick] (1.9,0.5) -- (2.1,0.5);
	\draw[thick] (1.9,3) -- (2.1,3);
	\draw[thick] (2,0.5) to node[auto,swap] {$\guar{P}$} (2,3) ;
	
	\draw[thick] (1.9,5) -- (2.1,5);
	\draw[thick] (1.9,6) -- (2.1,6);
	\draw[thick] (2,5) to node[auto,swap] {$\guar{\compl{P}}$} (2,6) ;
	\draw[thick] (2.9,0.5) -- (3.1,0.5);
	\draw[thick] (2.9,5) -- (3.1,5);
	\draw[thick] (3,0.5) to node[auto,swap] {$\safety{P}$} (3,5) ;
	\draw[dotted] (0,0.5) -- (3,0.5);
	\draw[dotted] (0,5) -- (3,5);
	\draw[dotted] (0,6) -- (2,6);
	\draw[dotted] (0,3) -- (2,3);
	\draw[dotted] (0,4) -- (1,4);
	\node at (0,0)  (si) {$\itraces$};
	\node at (-4,4) (ugly) {ugly prefixes};
	\draw[->] (ugly) to[bend right] (-1.8,3.8);
	\node at (-4,1.5) (good) {good prefixes};
	\draw[->] (good) to[bend right] (-1.2,0.8);
	\draw[->] (good) to[bend right] (-1.2,1.8);
	\node at (-4,5.5) (bad) {bad prefixes};
	\draw[->] (bad) to[bend right] (-1.2,5.3);
		\node at (-5,0) (s0) {$\Sigma^0$};
		\draw[->] (s0) to node[auto, swap] {prefix length} (si);
\end{tikzpicture}
\end{center}
\caption{The connection between Kupferman {\&} Vardi and Bauer's prefix characterisations and our kernel and closure operators}
\label{fig:operators}
\end{figure*}
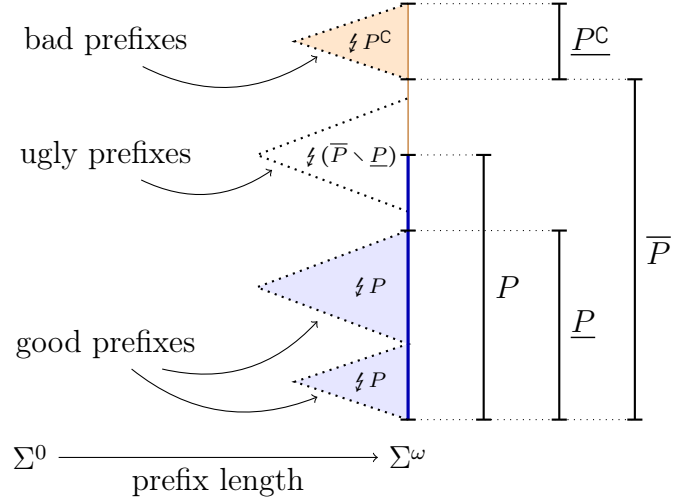

\subsubsection{Ugly prefixes}
Bauer et al.~\cite{bauercomparing} further define \emph{ugly prefixes} as those finite prefixes that cannot be finitely extended into good nor bad prefixes. Note that the good, bad, and ugly prefixes do not constitute a complete classification of all finite prefixes. For example, the finite prefix $\textit{ppp}\dots$ is neither good nor bad for $p \Until q$, but it is not ugly either, as it can be extended with $q$ giving a good prefix, or with $\emptyset$ giving a bad prefix. Here, $\emptyset$ is the state satisfying neither $p$ nor $q$.

\begin{thm}[Frontier view of ugly traces]\label{thm:frontierugly}
The ugly prefixes of a property $P$ are the finite definitive prefixes of the unmonitorable traces of $P$:
$$t\ \text{is ugly} \; \longleftrightarrow \; t \in \dclose (\safety{P}\setminus\guar{P}) \cap \ftraces$$
\end{thm}
\begin{proof}
Showing each direction separately:
\begin{description}
	\item[$\implies$] If $t$ is ugly then all finite extensions of $t$ are not definitive for $P$ or $\compl{P}$. Let $w$ be any infinite extension of $t$. If $w \in \guar{P} \cup \guar{\compl{P}}$, then there must be some finite prefix of $w$ that is definitive for $P$ or for $\compl{P}$. As such a prefix cannot exist,  $w$ is in neither $\guar{P}$ nor $\guar{\compl{P}}$. As $w$ is arbitrary, all infinite extensions of $t$ are in $\compl{(\guar{P} \cup \guar{\compl{P}})} = \safety{P}\setminus\guar{P}$. Thus $t \in \dclose (\safety{P}\setminus\guar{P}) \cap \ftraces$.
	\item[$\impliedby$] Given a finite prefix $t \in \dclose (\safety{P}\setminus\guar{P})$, we must show that $t$ is ugly. Let $u \in \extensions t$ be any finite extension of $t$. As definitive sets are closed under extension, $u \in \dclose (\safety{P}\setminus\guar{P})$. As the guarantee kernel of $P$, the guarantee kernel of $\compl{P}$, and the frontier $\safety{P}\setminus\guar{P}$ are all disjoint, and $u$ is definitive for the frontier, $u$ cannot be definitive for $P$ (i.e.~good) nor for $\compl{P}$ (i.e~bad). Therefore, $t$ has no good nor bad finite extensions, and thus $t$ is an ugly prefix.
\popQED
\end{description}

\end{proof}

\subsubsection{BLS monitorability}

Bauer et al.~\cite{bauercomparing} defines $P$ as monitorable iff there are no (finite) ugly prefixes. Intuitively, if a property is BLS (Bauer-Leucker-Schallhart) monitorable, then no matter what prefix has already been observed, a definitive answer is always possible, in the sense of a liveness property.
 Due to Theorem~\ref{thm:frontierugly}, this definition is equivalent to the following:
\begin{definition}[BLS monitorability]
$$P\ \text{is BLS monitorable} \; \longleftrightarrow \; \guar{\safety{P} \setminus \guar{P}} = \emptyset$$
Equivalently, $P$ is BLS monitorable iff the unmonitorable traces of $P$ are a morbidity property. Topologically, $P$ is BLS monitorable iff the \emph{frontier} of $P$ has an \emph{empty interior}.\footnote{Because the frontier of $P$ is always closed, we could here say ``nowhere dense'' instead of ``empty interior''.} This topological characterisation was first established by Diekert and Leucker~\cite{diekert2014}, and is independently reconstructed here.\end{definition}
\begin{thm}
	Every BLS monitorable property is weak monitorable.
\end{thm}
\begin{proof}
	Let $P$ be BLS monitorable, i.e., $\guar{\safety{P} \setminus \guar{P}} = \emptyset$. If $P$ were not weak monitorable, then $\safety{P} \setminus \guar{P} = \itraces$. As $\guar{\itraces} = \itraces \neq \emptyset$, this is a contradiction. Thus $P$ must be weak monitorable.
\end{proof}
\begin{counterexample}[$\mathsf{Weak} \supsetneq \mathsf{BLS}$]
	Consider $\Always \Eventually a \lor \Next b$. The prefix $bbbb$ is good, so
        this property is weak monitorable, however it is not BLS monitorable as the prefix $aaaa$ is ugly.
\end{counterexample}
\begin{thm}[Closure properties for BLS monitorability]\label{thm:bauerclosure} BLS monitorability is closed under complement, finite intersection and finite union.
\end{thm}
\begin{proof}
\textit{Complement}: Let $P$ be a BLS monitorable property, i.e.\ $\guar{\safety{P} \setminus \guar{P}} = \emptyset$. Then $\guar{\safety{\compl{P}} \setminus \guar{\compl{P}}} = \guar{\compl{\guar{\compl{{\compl{P}}}}} \cap \compl{\guar{\compl{P}}}} = \guar{\compl{\guar{P}} \cap \compl{\guar{\compl{P}}}} = \guar{\compl{\guar{P}} \cap \safety{P}} = \guar{\safety{P} \setminus \guar{P}} = \emptyset$.\vspace{2ex}

\noindent \textit{Intersection}:  Let $P$ and $Q$ be BLS monitorable properties and let $w$ be a finite prefix. We must show that $w$ is not ugly for $P \cap Q$.
In other words, we must show that $w$ has a finite extension $wu$ that is either good or bad for $P \cap Q$.

As $P$ is BLS monitorable, $w$ is not ugly for $P$. Hence there exists a finite extension $wv$ that is either good or bad for $P$. 
If $wv$ is bad for $P$, it is also bad for $P \cap Q$, and we can take $u := v$ and be done. So assume $wv$ is good for $P$. As $Q$ is BLS monitorable, $wv$ has a finite extension $wvx$ which is either good or bad for $Q$. If $wvx$ is 
bad for $Q$, then it is also bad for $P \cap Q$, and we can take $u := vx$ and be done. So assume $wvx$ is good for $Q$. As $wvx$ is an extension of $wv$, it is also good for $P$. Thus it is good for $P \cap Q$ and we can take $u := vx$. \vspace{2ex}

\noindent  \textit{Union}: Because BLS monitorability is closed under intersection and complement, it is also closed under union by de Morgan's laws.
\end{proof}
\begin{counterexample}[Failure of closure over infinite intersection]
	Consider the following family of properties: $$P_i = \{ t \in \itraces \mid t\ \text{contains at least}\ i\ \text{states labelled}\ a\}$$
	Each $P_i$ is BLS monitorable, as appending $i$ $a$-states to any finite prefix will yield a good prefix. The infinite intersection $\bigcap_i P_i$, however, requires an $a$-state to occur infinitely often, which is not BLS monitorable.
\end{counterexample}
\begin{thm}\label{thm:safetyguaranteebauer}
All safety and guarantee properties are BLS monitorable.\footnote{This also follows from a well known theorem about Hausdorff spaces that the frontier of any closed or open set has an empty interior.}
\end{thm}
\begin{proof}
	Let $P$ be a safety property, i.e.\ $\safety{P} = P$. Then 
	$$\begin{array}{lclr}
\guar{\safety{P}\setminus\guar{P}} &=& \guar{{P}\setminus\guar{P}}&\text{($P$ is safety prop.)} \\[0.4em]
    & =& \guar{P\cap\compl{\guar{P}}}\\[0.4em]
        & =& \guar{P}\cap\guar{\compl{\guar{P}}} &\quad \text{(Theorem~\ref{thm:gkint})}\\[0.4em]
        & \subseteq & \guar{P}\cap{\compl{\guar{P}}} & \text{(g.k.~is subset)}\\[0.4em]
        & = & \emptyset
\end{array}$$
For guarantee properties, we take the complement.
\end{proof}
\begin{counterexample}[$\mathsf{BLS} \supsetneq \mathsf{Safety} \vee \mathsf{Guarantee}$]\label{cte:psi}
	Consider the following property $X$, originally due to Bauer et al.~\cite{bauercomparing}: $$(p \vee q) \Until r \vee \Always p$$
	 This property is neither a safety property (the trace $q^\omega$ violates $X$ but $q^\omega \in \safety{X}$, i.e.~this violation is not observable from a finite prefix) nor a guarantee property (the trace $p^\omega$ satisfies $X$ but $p^\omega \notin \guar{X}$, i.e.~this satisfaction is not observable from a finite prefix). However, this property is BLS monitorable, as any finite trace can be extended to one that satisfies or violates the property.
\end{counterexample}

\subsubsection{Obligation Properties}
Counterexample~\ref{cte:psi} above is neither a safety nor a guarantee property, but it is the
\emph{union} of a safety and a guarantee property. Manna {\&} Pnueli~\cite{safetyprogress0} define \emph{obligation properties} as
Boolean combinations of safety and guarantee properties:
\begin{definition}[Obligation Properties, per~\cite{safetyprogress0}]
Inductively:
\begin{itemize}
	\item All guarantee properties are obligation properties.
	\item All safety properties are obligation properties.
	\item If $A$ and $B$ are obligation properties, then so is $A \cap B$.
	\item If $A$ and $B$ are obligation properties, then so is $A \cup B$.
\end{itemize}
\end{definition}
\begin{thm} All obligation properties are BLS monitorable.	
\end{thm}
\begin{proof}
Corollary of Theorem~\ref{thm:safetyguaranteebauer} and our closure properties in Theorem~\ref{thm:bauerclosure}.
\end{proof}
\begin{counterexample}[$\mathsf{BLS} \supsetneq \mathsf{Obligation}$]
   Let $X$ be any property over $\Sigma = \{ a, b \}$ that is not an obligation property.
   $X'$ is the property over $\Sigma = \{ a , b, c \}$ given by:
   $$
   X' = X \cup \{ t \mid c\ \text{occurs in}\ t\}
   $$ 
   $X'$ is BLS monitorable, as any finite prefix $w$ can be extended with just one $c$ 
   to give a good prefix for $X'$ and is therefore non-ugly. But, $X'$ cannot be an obligation property, because if it were, we could show that $X' \setminus \{ t \mid~\!\!c\ \text{occurs in}\ t\} = X$ is an obligation property, which is a contradiction. 
\end{counterexample}
\subsubsection{Strong monitorability}
BLS monitorability allows for a property to have unmonitorable traces, but they must not be inevitable from any finite prefix. We could imagine a stronger kind of monitorability, which demands that the property contains no unmonitorable traces at all:
\begin{definition}[Strong monitorability]
$$P\ \text{is strong monitorable} \; \longleftrightarrow \; \underline{P} \cup \underline{P^C} = \Sigma^{\omega} \; \longleftrightarrow \; \safety{P}\setminus\guar{P} = \emptyset$$ 
Topologically, the strong monitorable properties are the sets with an empty frontier. 
Intuitively, a property $P$ is strong monitorable if for \emph{all} traces that satisfy or violate $P$, the satisfaction or violation of $P$ can be determined after only a finite observation. This corresponds to the \emph{complete monitorability} of Aceto et al.~\cite{aceto2019adventures}.

\end{definition}\begin{thm}
  Strong monitorable properties are the intersection of safety properties and guarantee properties.\footnote{This is an instance of the general topology theorem that all sets with an empty frontier are clopen.}
\end{thm}
\begin{proof}
	Let $P$ be a strong monitorable property, that is, $\safety{P}\setminus\guar{P}=\emptyset$. Then
	$\safety{P} = \safety{P} \setminus \guar{P} \cup \guar{P} = \guar{P}$. As $\guar{P} \subseteq P \subseteq \safety{P}$, this means that $\safety{P} = \guar{P} = P$. Therefore $P$ is both a safety and guarantee property. 
	Similarly let $P$ be both a safety and a guarantee property, i.e.~$\safety{P} = \guar{P} = P$. Therefore clearly $\safety{P} \setminus \guar{P} = \emptyset$, and thus $P$ is strong monitorable.
\end{proof}
As a corollary, strong monitorable properties are closed under complement, intersection, and union. 
\section{Related Work}
\label{sec:disc}
\subsection{Finfinite Semantics}
Aceto et al.~\cite{aceto2021} define monitoring systems essentially as two extension-closed predicates on $\traces$, $\textbf{acc}$ and $\textbf{rej}$. Because they are extension closed, judgements are \emph{irrevocable}, in the sense that if a monitor \textbf{acc}epts a finite trace, it must \textbf{acc}ept all extensions of that trace (and similarly for \textbf{rej}). This is highly reminiscent of the two sets we use for our \ltlthree semantics (which are definitive and therefore also extension-closed), however these monitors are in the context of \emph{finfinite semantics}, where properties can consist of both finite and infinite traces. Falcone et al.~\cite{Falcone2012} provide definitions of monitorability in a finfinite setting which they claim correspond to safety and guarantee properties, but, as Aceto et al.~\cite{aceto2021} show, this claim is false. Aceto et al.~\cite{aceto2021} then present a hierarchy of monitorability classes that, at first glance, highly resembles that of Figure~\ref{fig:moncube}, but finfinite semantics differ in significant ways from our infinite trace semantics. In a finfinite context, the property $\Next \bot$, using a weak interpretation of the $\Next$ operator, is not empty. Similarly, the property $\Next \top$, using a strong interpretation of the $\Next$ operator, is not trivially true. Under finfinite semantics, monitors cannot make the progress assumption that a given finite observation will eventually be extended with another state. This significantly changes the characteristics of our monitorability classes. For example, as Aceto et al.~\cite{aceto2019adventures} show, under finfinite semantics the only strong monitorable properties (what Aceto et al.~refer to as \emph{complete monitorability}) are $\emptyset$ and $\traces$, but in infinite semantics there are non-trivial strong monitorable properties. Similarly, as noted by  Aceto et al.~\cite{aceto2021}, when $\Sigma$ contains only one state $\sigma$, the property $\{\sigma^\omega\}$, which would be trivially true and thus strong monitorable in our infinite semantics, is not even weak monitorable (what Aceto et al.~\cite{aceto2021} call \emph{$\exists$PZ-monitorable}) under finfinite semantics. We conjecture that finfinite semantics form another dimension to our lattice in Figure~\ref{fig:moncube}, where the finfinite semantics classes are generally strictly smaller than the infinite semantics classes, but we leave detailed examination of finfinite semantics for future work. 
\subsection{Syntactic Treatments}

The topological characterisations developed in this work are purely semantic and apply to any property $L \subseteq \Sigma^\omega$. A natural progression for future work is to consider the syntactic characterisation of these monitorable fragments when restricted to the class of $\omega$-regular languages, analogous to the characterisation of the safety-progress hierarchy classes in LTL~\cite{safetyprogress0} or the work of Aceto et al.~\cite{aceto2019adventures} on monitorability over finfinite semantics using subsets of the modal $\mu$-calculus. 

Such syntactic treatments would facilitate the development of algorithms for the effective construction of monitors for specific formulae, and would allow for evaluating the completeness of existing symbolic monitoring algorithms against our semantic definitions. Early investigations into monitor construction for LTL used coinductive techniques~\cite{sen2003} and term-rewriting~\cite{geilen2001} to bridge the gap between finite observations and infinite requirements. More recently, Francalanza and Cini~\cite{francalanza2021} introduced a local proof system for explainable monitoring based on dual-sided proof relations ($s \vdash^+ \varphi$ and $s \vdash^- \varphi$), which have a clear semantic analogue in our answer-indexed families. Such symbolic approaches have also been extended to real-time properties in Metric Temporal Logic \cite{lima2023}.

\subsection{Alternative Property Classes}
Havelund {\&} Peled~\cite{Havelund2018} provide a different set of property classes to the conventional hierarchy. They define six classes, \textsf{AFR, AFS, SFR, SFS, NFR}, and  \textsf{NFS}, indicating if the property is \textsf{A}lways, \textsf{S}ometimes, or \textsf{N}ever \textsf{F}initely \textsf{R}efutable or \textsf{S}atisfiable. Safety properties coincide with \textsf{AFR}, guarantee properties with \textsf{AFS}, liveness properties with \textsf{NFR}, and morbidity properties with \textsf{NFS}. Good monitorable properties are $\mathsf{SFS} \cup \mathsf{AFS}$, and bad monitorable properties are $\mathsf{SFR} \cup \mathsf{AFR}$. The classes of weak monitorable properties and of strong monitorable properties are both Boolean combinations of other classes, so they can also be expressed in terms of Havelund {\&} Peled's framework similarly. BLS monitorability does not cleanly map into Havelund {\&} Peled's framework, and indeed they show that several of their classes contain both BLS monitorable and BLS unmonitorable properties. Havelund {\&} Peled do not characterise any monitorability class other than BLS monitorability.

\subsection{Revocable Semantics}

As previously mentioned, all \emph{obligation} properties~\cite{safetyprogress0} are BLS monitorable. If we move further up the safety-progress hierarchy~\cite{safetyprogress0,safetyprogress}, however, we quickly encounter completely unmonitorable properties, such as our \emph{request/acknowledge} format example $\Always (r \Rightarrow \Eventually a)$. 

Bauer et al.~\cite{bauergbu, bauercomparing} propose relaxing the irrevocability requirement of {\ltlthree} by giving a dialect of LTL, called RV-LTL, where answers may be changed as the observed prefix is extended with more states. This allows answers to be given even for unmonitorable properties, although with no guarantees of soundness for presumptive answers. More precisely, RV-LTL is an ad-hoc layering of \ltlthree on top of Pnueli's LTL for finite traces (here notated $\models_\textsf{F}$). Where \ltlthree would give the $\unknown$ answer, RV-LTL instead gives a \emph{presumptive} answer ($\top^\textsf{p}$ or $\bot^\textsf{p}$ ) based on the answer obtained from Pnueli's finite LTL:
\[
[u \models \varphi]_\textsf{RV} = \begin{cases}
 	\top & \text{if}\ [ u \models \varphi]_3 = \top \\
 	\bot & \text{if}\ [ u \models \varphi]_3 = \bot \\
 	\top^\textsf{p} & \text{if}\ [ u \models \varphi]_3 = \texttt{?}\ \text{and}\ u \models_\textsf{F} \varphi \\
 	\bot^\textsf{p} & \text{if}\ [ u \models \varphi]_3 = \texttt{?}\ \text{and}\ u \nmodels_\textsf{F} \varphi \\
 \end{cases}
\]
Intuitively, after a finite prefix $u$, a definitive answer ($\top$ or $\bot$) is unchangeable no matter how the prefix is extended, whereas a presumptive answer  ($\top^\textsf{p}$ or $\bot^\textsf{p}$ ) only applies if execution is stopped at that point. 

If the property in question is a safety property, then the only presumptive answer possible is $\top^\textsf{p}$, and likewise for guarantee properties and $\bot^\textsf{p}$. This means that for properties at the bottom of the safety-progress hierarchy~\cite{safetyprogress0}, \ltlthree is sufficient, as the single $\unknown$ answer can be interpreted as $\top^\textsf{p}$ or $\bot^\textsf{p}$ respectively. However, as noted by Bauer et al.~\cite{bauercomparing}, for an obligation property such as Counterexample~\ref{cte:psi}, which is neither a safety nor guarantee property, both $\top^\textsf{p}$ and $\bot^\textsf{p}$ answers are possible.

Like \ltlthree previously, the semantics of RV-LTL is presented only in terms of other logics. We believe that an inductive semantics can be designed along similar principles to that of \ltlthree given in the present paper, where our answer indexed-families instead produce four sets, two of which are definitive, rather than the two definitive sets we provide for \ltlthree.

As noted by O'Connor {\&} Wickstr\"om~\cite{quickstrom}, Pnueli's finite LTL is a logic of finite \emph{completed} traces, so the decision to judge partial traces as completed for the purpose of giving presumptive answers in RV-LTL is ad-hoc and can produce seemingly arbitrary answers for properties higher in the safety-progress hierarchy.  
For example, consider a system where a flashing light consistently alternates between \textsf{On} and \textsf{Off} states:
$$ \textsf{On}\ \textsf{Off}\ \textsf{On}\ \textsf{Off}\ \cdots $$
A simple property that we might wish to monitor for this system is that the light is $\textsf{On}$ infinitely often:
$$ \Always \Eventually \textsf{On} $$%
As this formula nests $\Always$ and $\Eventually$ operators, it is definitive in neither positive nor negative cases and will only give presumptive answers.  But the
presumptive answer given in RV-LTL depends only on the very last observed status of the light. For a trace where the light continuously alternates off and on, as above, we might intuitively say that presumptive answer ought to be true, but this formula would be considered presumptively false if our observation happens to end in a state where the light is off. Thus, the truth value obtained for this formula is rather sensitive to the point at which our finite observation ceases. 

\subsection{Branching Time Monitoring}
The vast majority of research on monitoring and run-time verification concerns \LT semantics, but we can also consider monitoring for branching-time properties. Aceto et al.~\cite{aceto2019adventures} compare monitors for \LT and branching-time semantics. Francalanza et al.~\cite{Francalanza20172,francalanza2017} provide an extensive investigation into branching time monitoring. Manolios {\&} Trefler~\cite{TreflerManolios2001} and Bouajjani et al.~\cite{Bouajjani1991}  examine safety and liveness properties analogously to Alpern {\&} Schneider~\cite{alpernschneider} but in a branching-time setting. Branching-time analogues of our monitorability classes are future work.

\section{Conclusion}
We have presented a comprehensive semantic foundation for run-time verification with linear-time infinite trace semantics. By identifying those prefixes that are definitive for a property, we have given an inductive, model-based semantic accounting of \ltlthree in terms of answer-indexed families of definitive sets, and in the process shown that \ltlthree is more accurately described as a more detailed presentation of conventional LTL, rather than a distinct logic in its own right. We have formalised the popular formula progression technique used in runtime verification and testing scenarios, and proved it sound and complete with respect to our semantics. Using definitive prefixes as a basis, we have elegantly characterised various classes of properties and their monitorability characteristics in a distinctly topological flavour, in the vein of Alpern {\&} Schneider~\cite{alpernschneider}. All of our work has been mechanised in over 2900 lines of Isabelle/HOL proof script.

Our theory of definitive prefixes might provide a semantic foundation for other logics of partial traces, such as the LTL$^\pm$ of Eisner et al.~\cite{ltlpm}, QuickLTL from O'Connor {\&} Wickstr\"om~\cite{quickstrom}, or the aforementioned RV-LTL~\cite{bauercomparing}. Our answer-indexed families may also be applicable to other multi-valued logics. Examples include rLTL~\cite{rltl} and the five-valued logic of Chai et al.~\cite{chai}. Our monitorability characterisations are likely to have analogues in other semantic contexts such as when monitoring for branching-time properties or hyperproperties. 

\nocite{*}
\bibliographystyle{elsarticle-num-names}
\bibliography{cites}
\end{document}